\documentclass[article]{llncs}

\usepackage{graphicx}
\usepackage{amssymb}
\usepackage{times}
\usepackage{wrapfig}

\usepackage{tikz}
\usepgflibrary{shapes}
\usepackage{pgf,pgfarrows,pgfnodes,pgfautomata,pgfheaps,pgfshade}

\renewcommand{\qed}{\hfill \ensuremath{\Box}}
\def\tb{\rhd} 

\title{New Results on Equilibria in Strategic Candidacy}

\author{J\'erome Lang\inst{1} \and Nicolas Maudet\inst{2} \and Maria Polukarov\inst{3} \and Alice Cohen-Hadria\inst{2}
\institute{LAMSADE, Universit\'e Paris-Dauphine, Paris, France\\  \email{lang@lamsade.dauphine.fr}  \and LIP6, Universit\'e Pierre et Marie Curie, Paris, France\\  \email{nicolas.maudet@lip6.fr }  \and University of Southampton, United Kingdom\\ \email{mp3@ecs.soton.ac.uk}
}
}

\begin{document}

\maketitle


\begin{abstract}
We consider a voting setting where candidates have preferences about the outcome of the election and are free to join or leave the election. The corresponding candidacy game, where candidates choose strategically to participate or not, has been studied 
by Dutta et al.~\cite{DuttaElECONOMETRICA2001}, who showed that no non-dictatorial voting procedure satisfying unanimity is candidacy-strategyproof, that is, is such that the joint action where all candidates enter the election is always a pure strategy Nash equilibrium. In \cite{DuttaElJET2002} Dutta et al. also showed that for some voting tree procedures, there are candidacy games with no pure Nash equilibria, and that for the rule that outputs the sophisticated winner of voting by successive elimination, all games have a pure Nash equilibrium.
No results were known about other voting rules. Here we prove several such results. 
For four candidates, the message is, roughly, that most scoring rules (with the exception of Borda)  do  not guarantee the existence of a pure Nash equilibrium but that Condorcet-consistent rules, for an odd number of voters, do. 
For five candidates,  most rules we study no longer have this guarantee.  Finally, we identify one prominent rule that guarantees the existence of a pure Nash equilibrium for any number of candidates (and for an odd number of voters): the Copeland rule.  We also show that under mild assumptions on the voting rule, the existence of strong equilibria cannot  be guaranteed. 
\end{abstract}



\section{Introduction}

A main issue for the evaluation of voting rules is their ability to resist various sorts of strategic behavior.
Strategic behavior can come from the voters reporting insincere votes ({\em manipulation}); from a third party, typically the chair, acting on the set of voters or candidates ({\em control}), on the votes ({\em bribery}, {\em lobbying}), or on the voting rule ({\em e.g., agenda control}).
However, strategic behavior by the {\em candidates} has received less attention 
than strategic behavior by the voters and (to a lesser extent) by the chair. One form thereof involves choosing optimal political platforms. But probably the simplest form {\em comes from the ability of candidates to decide whether to run for the election or not}, which is the issue we address here. The following table summarizes this rough classification of strategic behavior in voting, according to the identity of the strategizing agent(s) and 
to another relevant dimension, namely what the strategic actions bear on---voters, votes or candidates (we omit the agenda to keep the table small). 

$$
\begin{tabular}{|c||c|c|c|}
\hline 
\begin{tabular}{c} {\bf actions} $\rightarrow$  \\ {\bf agents} $\downarrow$ \end{tabular} & {\em voters} & {\em votes} & {\em candidates} \\ \hline  \hline
{\em voters} & \begin{tabular}{c} strategic participation \end{tabular} & manipulation & - \\ \hline
{\em third party / chair} & voter control & \begin{tabular}{c} bribery, lobbying \end{tabular} & \begin{tabular}{c} candidate  control, \\ cloning  \end{tabular}\\ \hline
{\em candidates} & - & - & \begin{tabular}{c} {\em strategic} {\em candidacy} \end{tabular} \\ \hline
\end{tabular}
$$
Strategic candidacy does happen frequently in real-life elections, both in large-scale political elections and in small-scale, low-stake elections ({\em e.g.}, electing a chair in a research group, 
or---moving a little bit away from elections---reputation systems). Throughout the paper we consider a finite set of {\em potential candidates}, which we simply call {\em candidates} when this is not ambiguous, and we make the following assumptions:

\begin{enumerate}
\item each candidate may choose to run or not for the election;
\item each candidate has a preference ranking over candidates;
\item each candidate ranks himself on top of his ranking; 
\item the candidates' preferences are common knowledge among them; 
\item the outcome of the election as a function of the set of candidates who choose to run is common knowledge among the candidates.
\end{enumerate}

With the exception of 3, these assumptions were also made in the original model of Dutta et al.~\cite{DuttaElECONOMETRICA2001} which we discuss below. Assumption 2 amounts to saying that a candidate is interested only in the winner of the election\footnote{In some contexts, candidates may have more refined preferences that bear for instance on the number of votes they get, how their score compares to that of other candidates, 
etc. We do not consider these issues here.} and has no indifferences or incomparabilities. Assumption 3 (considered as optional in~\cite{DuttaElECONOMETRICA2001}) is a natural domain restriction in most contexts. 
Assumptions 4 and 5 are common game-theoretic assumptions; 
note that we do not have to assume that the candidates know precisely how voters will vote, nor even the number of voters; they just have to know the choice function mapping every subset of candidates to a winner. Assumption 4 is required only when strong Nash equilibria are considered. 

Existing work on strategic candidacy is rather scarce. 
Dutta {\em et al.} \cite{DuttaElECONOMETRICA2001,DuttaElJET2002} formulate the strategic candidacy game and prove that no non-dictatorial voting procedure satisfying unanimity is candidacy-strategyproof (or equivalently, that for any non-dictatorial voting procedure satisfying unanimity, there is a profile for which the joint action where all candidates enter the election is not a pure Nash equilibrium). Then, Dutta {\em et al.} \cite{DuttaElJET2002} exhibit a (non-anonymous) voting tree rule 
for four candidates for which there is a candidacy game with no pure Nash equilibria. They also show that for the voting rule that outputs the sophisticated outcome for voting by successive elimination, the existence of a pure Nash equilibrium is guaranteed.
Some of these results are discussed further (together with simpler proofs) by Ehlers and Weymark \cite{EhlersWeymark03}, and extended to voting correspondences by Ereslan \cite{Ereslan04} and Rodriguez \cite{Rodriguez06a}, and to probabilistic voting rules by Rodriguez \cite{Rodriguez06b}. Brill and Conitzer \cite{BrillConitzerAAAI2015} extend the analysis to also include strategic behavior by the voters. Polukarov {\em et al.} \cite{PolukarovOR15} study equilibrium dynamics in candidacy games, in which candidates may strategically decide to enter the election or withdraw their candidacy. Obraztsova {\em et al.} \cite{ObraztsovaEMR15} study strategic candidacy games with lazy candidates, whose utility function results form the outcome of the election minus a small penalty for running for election.

Studying the equilibria of a candidacy game helps predicting the set of actual candidates and therefore the outcome of the vote.
However, little is known about this: we only know that for any reasonable voting rule, there are some candidacy games for which the set of all candidates is not a pure Nash equilibrium, that there exist candidacy games with no pure Nash equilibria, and that a specific rule, defined from the successive elimination procedure and assuming that voters reason by backward induction, all candidacy games have a pure Nash equilibrium. We do not know, for instance,
whether pure Nash equilibria always exist 
for common voting rules such as plurality, Borda or Copeland.

In this paper, we go further in this direction and prove some positive as well as some negative results. We first consider the case of {\em four} candidates and show that for an odd number of voters, a pure Nash equilibrium always exists for Condorcet-consistent rules, while for most scoring rules and as well as for single transferable vote and plurality with runoff, this is not the case. 
Over the {\em five} candidate frontier, we know very few rules that,
can still guarantee the existence of an equilibrium. 
We show that for the
Copeland rule, and an odd number of voters, there is always a pure Nash equilibrium, whichever the number of candidates. On the negative side, we show that for most scoring rules, for at least four candidates, and for Borda, maximin, and the uncovered set for at least five candidates, and for the top cycle with at least seven candidates, there are candidacy games without pure Nash equilibria.  We also prove a simple impossibility theorem showing that {\em strong} Nash equilibria are not guaranteed to exist provided the voting rule satisfies two mild conditions satisfied by most common rules.


The paper unfolds as follows. In Section \ref{sec:model} we define the strategic candidacy games and give a few preliminary results. In Section \ref{sec:4candidates}  we focus on the case of four candidates. The case of five candidates is considered in Section \ref{sec:5candidates}. Section \ref{sec:more} deals with candidacy games with more candidates.  
In Section \ref{sec:control} we discuss strong Nash equilibria,
and relate the candidacy game to candidate control. Finally, in Section \ref{sec:further} we discuss further issues.



\section{Model and preliminaries}\label{sec:model}
In this section, we define the  strategic candidacy model, show that it induces a normal form game, and give preliminary results on the existence of Nash equilibria.


\subsection{Voting rules}\label{subsec:rules}

Let $X = \{x_1, x_2, \dots x_m\}$  be a set of \emph{potential} candidates and $N = \{1, 2, \dots n\}$ a set of $n$ voters.
{\em We assume $n$ is odd}, so that pairwise majority ties do not occur. While this is a mild assumption when the number of voters is large, this implies a loss of generality for some of our results, and when this is the case we will make it clear.


For any subset $Y \subseteq X$ of the candidates, a $Y$-vote is a linear ordering over $Y$. A $Y$-profile $P = \langle \succ_1, \ldots, \succ_n\rangle$ is a collection of $n$ $Y$-votes. Although voting rules are often defined for a fixed set of candidates, here we define them for an arbitrary subset of the set of potential candidates: a (resolute) {\em voting rule} maps every $Y$-profile, for every $Y \subseteq X$, to a candidate in $Y$.
We will only consider resolute rules; we will first define their irresolute version and then assume that ties are broken up according to a fixed priority relation over the candidates.  Because voting rules are applied to varying sets of candidates, we assume that the tie-breaking rule is defined as a linear ordering on the whole set of potential candidates $X$, and projected to subsets of candidates: if $x$ has priority over $y$ (noted $x \rhd y$) when all potential candidates run, this will still be the case for any set of candidates that contains $x$ and $y$. 

We now define the rules we will use in the paper. (For each of them we define its irresolute version, its resolute version being obtained as explained above.)

\paragraph{Scoring rules.}
A \emph{scoring rule} (for a varying set of candidates) is defined by a collection of vectors $\vec{S}_m = \langle s_1, \dots s_m \rangle$ for all $m$, with $s_1 \geq s_2 \geq \ldots \geq s_m$ and $s_1 > s_m$. For each $m$ and each $i \leq m$, $s_i$ is the number of points obtained by the candidate ranked in position $i$, and the winning candidate(s) maximizes the sum of points obtained from all votes. Formally speaking, defining a family of scoring rules requires to specify a scoring vector for each size of a candidate set (for instance, $\langle 3, 1, 0 \rangle$ for three candidates, $\langle 4, 3, 2,0 \rangle$ for four candidates and so on). However, for the following classical rules, these collections of vectors are defined in a natural way: 
\begin{itemize}
\item {\em plurality}: $\vec{S}_m = \langle 1, 0, \dots 0 \rangle$;
\item {\em veto} (or {\em antiplurality}): $\vec{S}_m = \langle 1, \dots 1, 0 \rangle$;
\item {\em Borda}: $\vec{S}_m = \langle m-1, m-2, \dots 1, 0 \rangle$. 
\end{itemize}

\paragraph{Condorcet-consistent rules.}
Let $P$ be a profile and $N_P(c,x)$ be the number of votes in $P$ who rank $c$ above $x$. The majority graph $m(P)$ associated with $P$ is the graph whose vertices are the candidates and containing an edge from $x$ to $y$ whenever $N_P(x,y) > \frac{n}{2}$ (we say that $x$ {\em beats} $y$ in $m(P)$, denoted by $x \to_P y$). 
Because $n$ is odd, $m(P)$ is a tournament, {\em i.e}, a complete asymmetric graph.
A candidate $c$ is a \emph{Condorcet winner} if $c \to_P y$ for all $y \neq c$. A voting rule $r$ is {\em Condorcet-consistent} if $r(P) = \{c\}$ whenever there is a (unique) Condorcet winner $c$ for $P$.

Given a profile $P$, the \emph{top cycle} $TC(P)$ 
is the smallest $S \subseteq X$ such that  for every $x \in S$ and $y \in X \setminus S$, $x \to_P y$ . 
The \emph{uncovered set} $UC(P)$ is the set $S \subseteq X$ of candidates such that for any  $c \in S$ and for any other candidate $x$, if $x \to_P c$ then there is some $y$ such that $c \to_P y$ and $y \to_P x$.  
The \emph{maximin} rule chooses the candidate(s) $c$ that maximize $\min_{x\in X\setminus\{c\}} N_P(c,x)$.
The {\em Copeland} rule chooses the candidate(s) $c$ that maximize 
$| \{x \in X | c \to_P x \}|$.

\paragraph{Rules based on iterative elimination of candidates.} 
{\em Plurality with runoff} proceeds in two rounds:
we first select the two candidates $x$ and $y$ with highest plurality scores and the second round chooses between them according to majority. {\em Single transferable vote} (STV) proceeds in $m-1$ rounds: at each round, the candidate with the lowest plurality score among the remaining candidates (using tie-breaking if necessary) is eliminated. 


\subsection{Strategic candidacy}\label{subsec:model}


In addition to voters' preferences over candidates, expressed by the voter profile $P$, we assume that each candidate $i$ too has a linear preference ordering $\succ_i$ over candidates. We  assume furthermore that the candidates' preferences are \emph{self-supported}---that is, each candidates rank herself at the top of her ranking. Let $P^X=(P^X_c)_{c\in X}$ denote the candidates' preference profile. 

We assume that voters are sincere; therefore, when the set of candidates running for election is $Y \subseteq X$, each voter $i$ reports the restriction of $\succ_i$ to $Y$ and the obtained profile, denoted by $P^{\downarrow Y}$, is the restriction of $P$ to $Y$. 


Given a fixed voter profile $P$, a \emph{voting rule} $r$ can be seen as mapping each
$Y\subseteq X$ to a winner $r(P^{\downarrow Y})$ in $Y$. 
We use the  notation $Y \mapsto_{P,r} x$, or more simply, $Y \mapsto x$ when there is no ambiguity, to denote that the outcome of rule $r$ applied to profile $P$ restricted to the subset of candidates  $Y \subseteq X$ is $x$.

Each voting rule $r$ induces a natural \emph{game form}, where the set of players is given by the set of potential candidates $X$, and the strategy set available to each player is $\{0,1\}$ with $1$ corresponding to entering the election and $0$ standing for  withdrawal of candidacy. A \emph{state} $s$ of the game is a vector of strategies $(s_c)_{c\in X}$, where $s_c \in \{0,1\}$. For convenience, we use $s_{-z}$ to denote $(s_c)_{c\in X\setminus\{z\}}$---i.e., $s$ reduced by the single entry of player $z$. Similarly, for a state $s$ we use $s_Z$ to denote the strategy choices of a coalition $Z \subseteq X$ and $s_{-Z}$ for the complement, and we write $s = (s_Z,s_{-Z})$. 

The outcome of a state $s$ is $r\left(P^{\downarrow Y}\right)$ where $c\in Y$ if and only if $s_c=1$.\footnote{When clear from the context, we use vector $s$ to also denote the set of candidates $Y$ that corresponds to state $s$;   
e.g., if $X = \{x_1, x_2, x_3\}$, we note $\{x_1,x_3\}$ and (1,0,1) interchangeably. 
}  Coupled with a voter profile $P$ and a candidate profile $P^X$, this defines a normal form game $\Gamma=\langle X, P, r, P^X\rangle$ with $m$ players. Here, player $c$ prefers outcome $\Gamma(s)$ over outcome $\Gamma(s')$ if ordering $P^X_c$ ranks $\Gamma(s)$ above $\Gamma(s')$.


\subsection{Game-theoretic concepts}\label{subsec:equilibria}
Having defined a normal form game, we can now apply standard game-theoretic solution concepts. Let $\Gamma=\langle X, P, r, P^X\rangle$ be a candidacy game, and let $s$ be a state in $\Gamma$. We say that a coalition $Z \subseteq X$ has an \emph{improving move} in $s$ if there is $s'_Z$ such that $\Gamma(s_{-Z},s'_Z)$ is preferred to $\Gamma(s)$ by every 
$z \in Z$. In particular, the improving move is \emph{unilateral} if $|Z|=1$. A state is a \emph{(pure strategy) Nash equilibrium} (NE) if it has no unilateral improving moves, and 
a $k$-NE if no coalition with $|Z| \leq k$ has an improving move. A \emph{strong Nash equilibrium (SE)} \cite{Aumann59} is a state with no improving moves. 
\begin{example}
Consider the game $\langle \{a,b,c,d\},P,r,P^X \rangle$, where $r$ is the Borda rule, and $P$  and $P^X$ are as follows:\footnote{In our examples, when the tie-breaking ordering is not specified it is assumed to be lexicographic.  We generally 
omit curly brackets. 
The first row in $P$ indicates the number of voters casting the different ballots. We use the common convention of writing votes vertically, with the topmost candidate being preferred.}
\begin{small}
$$
\begin{array}{cc}
P & P^X\\
\begin{array}{|ccccccc|}
\hline
1 & 1 & 1 & 1 & 1 & 1 & 1 \\
\hline
b & c & c & a & d & b & a \\
d & d & d & c & a & c & b\\
a & a & b & b & c & d & c\\
c & b & a & d & b & a & d\\
\hline
\end{array} 
&
\begin{array}{|cccc|}
\hline
a & b & c & d \\
\hline
a & b & c & d \\
d & a & b & a \\
b & d & a & c \\
c & c &  d & b \\
\hline
\end{array}
\end{array}
$$
\end{small}
\noindent The state (1,1,1,1) is not an NE: $abcd \mapsto c$, but $abc \mapsto a$, and $d$ prefers $a$ to $c$, so for $d$, leaving is an improving move. Now, (1,1,1,0) is an NE, as noone has an improving move neither by joining ($d$ prefers $a$ over $c$), or by leaving (obviously not $a$; if $b$ or $c$ leaves then the winner is still $a$). It can be checked that this is also an SE.
\end{example}


\subsection{Preliminary results}\label{subsec:preliminaries}

Regardless of the number of voters and the voting rule used, a straightforward observation is that a candidacy game with {\em three} candidates is guaranteed to possess an NE. Note that this does not hold for SE.\footnote{Here is a counterexample (communicated to us by Markus Brill).
The selection rule is $abc \mapsto b$; $ab \mapsto a$; $ac \mapsto c$; $bc \mapsto c$; it can be easily implemented by the scoring rule with scoring vector $(5,4,0\rangle$ with 5 voters. Preferences of candidates are: $a: a \succ b \succ c; b: b \succ c \succ a; c: c \succ a \succ b$. The group deviations are: in $\{a,b,c\}$, $c$ leaves;  in $\{a,b\}$, $b$ leaves and $c$ joins; in $\{a,c\}$, $b$ joins; in $\{b,c\}$, $a$ joins; in $\{a\}$, $c$ joins; in $\{b\}$, $c$ joins; in $\{c\}$, $a$ and $b$ join.}

The first question which comes to mind is whether examples showing the absence of NE transfer to larger set of candidates. They indeed do, under an extremely mild assumption. 
We say that a voting rule is \emph{insensitive to bottom-ranked candidates} (IBC) if given any profile $P$ over $X = \{x_1, \ldots, x_m\}$, if $P'$ is the profile over $X \cup \{x_{m+1}\}$ obtained by adding $x_{m+1}$ at the bottom of every vote of $P$, then $r(P') = r(P)$. This property is extremely weak (much weaker than Pareto efficiency) and is satisfied by almost all `common' voting rules. 

\begin{lemma}\label{lemma-more-candidates}
For any voting rule $r$ satisfying IBC, if there exists $\Gamma = \langle X, P, r, P^X\rangle$ with no NE, then there exists $\Gamma' = \langle X', P', r, P^Y\rangle$ with no NE, where $|X'| = |X| + 1$.
\end{lemma}
\begin{proof}
Take  $\Gamma$ with no NE, with $X = \{x_1, \ldots, x_m\}$. Let $X' = X \cup \{x_{m+1}\}$, $P'$ the profile obtained from $P$ by adding $x_{m+1}$ at the bottom of every vote, and $P^{X'}$ be the candidate profile obtained by adding $x_{m+1}$ at the bottom of every ranking of a candidate $x_i$, $i < m$, and whatever ranking for $x_{m+1}$. Let $Y \subseteq X$. Because $Y$ is not an NE for $\Gamma$, some candidate $x_i \in X$ has an interest to leave or to join, therefore $Y$ is not an NE either for $\Gamma'$. Now, consider $Y' = Y \cup \{x_{m+1}\}$. If $x_i \in X$ has an interest to leave (resp., join) $Y$, then because $r$ satisfies IBC, the winner in $Y' \setminus \{x_i\}$ (resp., $Y' \cup \{x_i\}$) is the same as in $Y \setminus \{x_i\}$ (resp., $Y \cup \{x_i\}$), therefore $x_i \in X$ has an interest to leave (resp., join) $Y'$,  therefore $Y'$ is not an NE.
\qed
\end{proof}

We will use this induction lemma to extend some of our negative results to an arbitrary number of candidates. A noticeable exception is the veto rule, which does not satisfy IBC. In Appendix A we provide a specific lemma to handle this rule. \medskip

The following result applies to any number of candidates and Condorcet-consistent rules. 
\begin{proposition}\label{prop:CondorcetWinner}
Let $\Gamma=\langle X, P, r, P^X\rangle$ be a candidacy game where $r$ is Condorcet-consistent. If $P$ has a Condorcet winner $c$ then for any $Y \subseteq X$, 
\begin{center}$Y$ is an SE $\Leftrightarrow$ $Y$ is an NE $\Leftrightarrow$ $c \in Y$.\end{center}
\end{proposition}
The very easy proof can be found in Appendix A. If $P$ has no Condorcet winner, the analysis becomes more complicated. We provide results for this more general case in the following sections.



\section{The first frontier: four candidates}\label{sec:4candidates}
With only four potential candidates, we exhibit a sharp contrast between the Condorcet-consistent rules, for which a Nash equilibrium is guaranteed to exist (for odd $n$), and many other voting rules.

\subsection{Scoring rules}
We make use of a powerful result by Saari~\cite{SaariJET1989} which states that for almost all scoring rules, any choice function can result from a voting profile. For four candidates~\cite{Saari1996}, we define a {\em Saari rule} as a rule for which, when the scoring vector for three candidates is of the form $\langle w_1, w_2, 0 \rangle$, then the vector for four candidates is {\em not} $\langle 3w_1, w_1 + 2w_2, 2w_2, 0 \rangle$. For instance, plurality and veto are Saari rules, but {\em the Borda rule is not a Saari rule}.
For any Saari rule, any choice function can result from a voting profile~\cite{SaariJET1989,Saari1996}. This means that our question boils down to check whether a \emph{choice function}, together with some coherent candidates' preferences, can be found such that no NE exists with four candidates. We solved this question by encoding the problem as an Integer Linear Program (ILP), the details of which can be found in Appendix B. 

It turns out that such choice functions do exist. We depict one of them in Figure~\ref{fig:ChoiceFunctionNoNe} (where arrows denote deviations and the right part of each cell denotes the winner), which rules out the existence of an NE when taken with the candidates preferences: 
$$
\begin{array}{ll}
a: & a \succ b \succ c \succ d\\
b: & b \succ a \succ c \succ d \\
c: & c \succ d \succ a \succ b \\
d: & d \succ a \succ b \succ c \\
\end{array}
$$

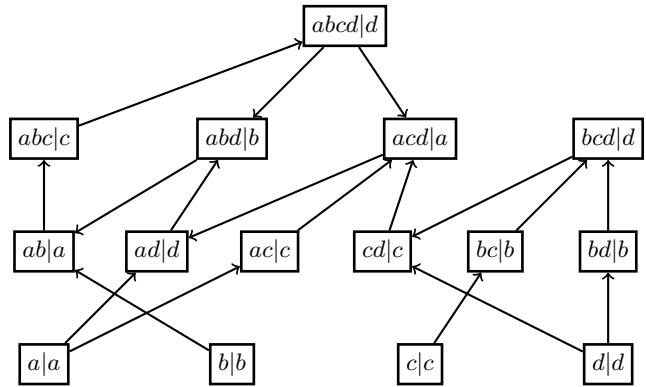
\begin{figure}
\begin{center}
\begin{tikzpicture}[scale=0.5]
\tikzstyle{is-in}=[rectangle,draw, line width=1pt]

\node (abcd) at (8,9) [is-in,label=above:\footnotesize{$$}, label=below:\footnotesize{}] {$abcd|d$};
\node (abc) at (0,6) [is-in,label=above:\footnotesize{$$}, label=below:\footnotesize{}] {$abc|c$};
\node (abd) at (5,6) [is-in,label=above:\footnotesize{$$}, label=below:\footnotesize{}] {$abd|b$};
\node (acd) at (10,6) [is-in,label=above:\footnotesize{$$}, label=above:\footnotesize{}] {$acd|a$};
\node (bcd) at (15,6) [is-in,label=above:\footnotesize{$$}, label=above:\footnotesize{}] {$bcd|d$};
\node (ab) at (0,3) [is-in,label=below:\footnotesize{$$}, label=below:\footnotesize{}] {$ab|a$};
\node (ac) at (6,3) [is-in,label=below:\footnotesize{$$}, label=below:\footnotesize{}] {$ac|c$};
\node (ad) at (3,3) [is-in,label=below:\footnotesize{$$}, label=below:\footnotesize{}] {$ad|d$};
\node (bc) at (12,3) [is-in,label=below:\footnotesize{$$}, label=below:\footnotesize{}] {$bc|b$};
\node (bd) at (15,3) [is-in,label=below:\footnotesize{$$}, label=below:\footnotesize{}] {$bd|b$};
\node (cd) at (9,3) [is-in,label=below:\footnotesize{$$}, label=below:\footnotesize{}] {$cd|c$};
\node (a) at (0,0) [is-in,label=below:\footnotesize{$$}, label=below:\footnotesize{}] {$a|a$};
\node (b) at (5,0) [is-in,label=below:\footnotesize{$$}, label=below:\footnotesize{}] {$b|b$};
\node (c) at (10,0) [is-in,label=below:\footnotesize{$$}, label=below:\footnotesize{}] {$c|c$};
\node (d) at (15,0) [is-in,label=below:\footnotesize{$$}, label=below:\footnotesize{}] {$d|d$};


\draw [thick, ->] (ab) -- (abc);
\draw [thick,->] (ad) -- (abd);

\draw [thick,->] (bc) -- (bcd);

\draw [thick,->] (ac) -- (acd); 

\draw [thick,->] (bd) -- (bcd); 

\draw [thick,->] (cd) -- (acd);

\draw [thick,->] (abc) -- (abcd);

\draw [thick,->] (acd) -- (ad);

\draw [thick,->] (abd) -- (ab);

\draw [thick,->] (bcd) -- (cd);

\draw [thick,->] (abcd) -- (acd);
\draw [thick,->] (abcd) -- (abd);

\draw [thick,->] (b) -- (ab);
\draw [thick,->] (a) -- (ad);
\draw [thick,->] (a) -- (ac);
\draw [thick,->] (c) -- (bc);
\draw [thick,->] (d) -- (cd);
\draw [thick,->] (d) -- (bd);

\end{tikzpicture}
\caption{A choice function without NE}
\label{fig:ChoiceFunctionNoNe}
\end{center}
\end{figure}

The following result then follows directly.

\begin{proposition}\label{saari4}
For four candidates, if $r$ is a Saari rule, there are candidacy games without Nash equilibria.  
\end{proposition}

As a corollary, we get that: 
\begin{corollary}\label{saari4}
For plurality, veto (and more generally, for $k$-approval with any $k$), there are candidacy games without NE.
\end{corollary}

Note that Saari's result shows that counter-examples can be obtained for all these scoring rules, but it does not directly provide the profile satisfying this choice function. 
These profiles may involve a large number of voters.
For plurality, we exhibit a profile with 13 voters corresponding to the choice function given in Fig.  \ref{fig:ChoiceFunctionNoNe}, whose preferences are shown on the left part of the table below.
The right part of the table represents $P^X$.

\begin{footnotesize}
\[
\begin{array}{|ccccccccc|}
\hline
3 & 1 & 1 & 1 & 1 & 1 & 1 & 2 & 2 \\
\hline
d & d & d & a & a & a & b & b & c \\
c & b & a & b & c & d & c & a & b \\
a & c & b & c & b & b & d & c & d \\
b & a & c & d & d & c & a & d & a \\
\hline
\end{array} 
\begin{array}{|cccc|}
\hline
a & b & c & d \\
\hline
a & b & c & d \\
b & a & d & a \\
c & c & a & b \\
d & d & b & c \\
\hline
\end{array}
\]
\end{footnotesize}
\noindent
Similar profiles can be obtained for other Saari rules.  As for the Borda rule, which is not a Saari rule, it stands as an exception:
\begin{proposition}
For Borda and $m=4$, every candidacy game has an NE. 
\end{proposition}
This result was obtained by a translation into an integer linear program, then run on a computer. It relies on the fact that Borda rule can be computed from the weighted majority graph, and by adding the corresponding constraints into the ILP (for the details of this ILP, see Appendix B). 
The infeasibility of the resulting set of constraints shows that no instances without NE can be constructed.  \medskip

However, it takes only coalitions of pairs of agents to ruin this stability.  Indeed, 
for Borda and $m=4$, there are candidacy games without 2-NE. 
This can be seen on the following candidacy game:   
\[
\begin{array}{|ccccc|}
\hline
1 & 1 & 1 & 1 & 1 \\
\hline
b & c & d & a & b  \\
d & d & a & b & c  \\
c & a & c & c & d  \\
a & b & b & d & a  \\
\hline
\end{array} 
\begin{array}{|cccc|}
\hline
a & b & c & d \\
\hline
a & b & c & d \\
c & a & a & b \\
d & c & d & a \\
b & d & b & c \\
\hline
\end{array}
\]
Only $s_1=(0,1,1,1)$ and $s_2=(1,1,0,1)$ are NE, with $bcd \mapsto b$, and $abd \mapsto d$. From $s_1$ the coalition $\{a,c\}$ has an improving move to $s_2$ as they both prefer $d$ to $b$. From $s_2$, if $b$ leaves and $c$ joins, they reach $(1,0,1,1)$, with $acd \mapsto c$ and both prefer $c$ to $d$. 

\subsection{Rules based on successive elimination}
Let us now focus on \emph{plurality with runoff} and \emph{single transferable vote}. 
For these  rules, it is no longer the case that any choice function can be implemented by such rules. For instance, for plurality with runoff, a necessary condition for the choice function to be implementable is that, for any subset of candidates $Y$, $|Y|\geq 3$, if $r(Y)=x$, then $x$ must win in pairwise comparison against \emph{some} candidate $y \in Y \setminus \{x\}$. For STV, a stronger condition is even required: for any subset of candidates $Y$, if $r(Y)=x$, it must be the case that $r(Z)=x$ for some set $Z\subset Y$ such that $|Z|=|Y-1|$. 

We make no claim that these conditions are sufficient to ensure a possible implementation.  However, by adding these constraints into our ILP, we generated a choice function that we could in turn implement with a specific profile, thus providing us the following result. 

\begin{proposition}\label{pro-stv4}
For plurality with runoff and single transferable vote and $m=4$, there are candidacy games without NE. 
\end{proposition}
\begin{proof}
We exhibit a counter-example with 9 voters.
The tie-breaking is $d\tb a \tb c  \tb b$. 

\begin{footnotesize}
\[
\begin{array}{|ccccccccc|}
\hline
1 & 1 & 1 & 1 & 1 & 1 & 1 & 1 & 1 \\
\hline
a & a & b & b & b & c & c & d & d \\
c & d & a & c & d & b & d & a & c \\
b & b & c & a & c & d & a & b & a \\
d & c & d & d & a & a & b & c & b \\
\hline
\end{array} 
\begin{array}{|cccc|}
\hline
a & b & c & d \\
\hline
a & b & c & d \\
b & a & a & a \\
d & c & b & c \\
c & d & d & b \\
\hline
\end{array}
\]
\end{footnotesize} \qed
\end{proof}

\subsection{Condorcet-consistent rules}
We now turn our attention to Condorcet-consistent rules. We recall that we assume the number of voters $n$ to be odd. 

\begin{proposition}\label{Condorcet-four} 
For $m=4$ (and $n$ odd), if $r$ is Condorcet-consistent then every candidacy game has an NE.
\end{proposition}

\begin{proof}

For any profile $P$, let $G_P$ the complete tournament obtained from the majority graph associated with $P$. Although we do not assume that $r$ is based on the majority graph, we nevertheless prove our result by considering all possible tournaments on four candidates (we shall get back to this point at the end of the proof). 
In the proof, when we speak of an ``NE in $G$'' we mean an NE in any candidacy game for which the profile $P$ is associated with the majority graph $G$.
There are four tournaments to consider (all others are obtained from these ones by symmetry). 
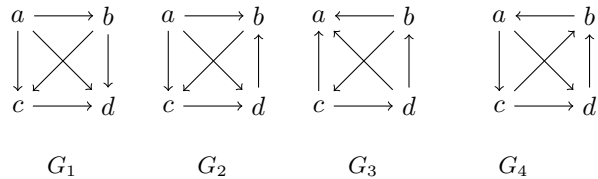
\begin{center}
\begin{tikzpicture}[scale=0.8]


\node (A1) at (0,2.5) [] {$a$};
\node (B1) at (1.5,2.5) [] {$b$};
\node (C1) at (0,1) [] {$c$};
\node (D1) at (1.5,1) [] {$d$};

\node (G1)     at (0.8,0) [text width=0.5cm,font=\footnotesize] {$G_1$};

\draw [->] (A1) -- (B1);
\draw [->] (A1) -- (C1);
\draw [->] (A1) -- (D1);
\draw [->] (B1) -- (D1);
\draw [->] (C1) -- (D1);
\draw [->] (B1) -- (C1);

\node (A2) at (2.5,2.5) [] {$a$};
\node (B2) at (4,2.5) [] {$b$};
\node (C2) at (2.5,1) [] {$c$};
\node (D2) at (4,1) [] {$d$};

\node (G2)     at (3.3,0) [text width=0.5cm,font=\footnotesize] {$G_2$};

\draw [->] (A2) -- (B2);
\draw [->] (B2) -- (C2);
\draw [->] (C2) -- (D2);
\draw [->] (D2) -- (B2);
\draw [->] (A2) -- (D2);
\draw [->] (A2) -- (C2);

\node (A3) at (5,2.5) [] {$a$};
\node (B3) at (6.5,2.5) [] {$b$};
\node (C3) at (5,1) [] {$c$};
\node (D3) at (6.5,1) [] {$d$};

\node (G3)     at (5.8,0) [text width=0.5cm,font=\footnotesize] {$G_3$};

\draw [->] (B3) -- (A3);
\draw [->] (C3) -- (A3);
\draw [->] (D3) -- (A3);
\draw [->] (B3) -- (C3);
\draw [->] (C3) -- (D3);
\draw [->] (D3) -- (B3);

\node (A4) at (8,2.5) [] {$a$};
\node (B4) at (9.5,2.5) [] {$b$};
\node (C4) at (8,1) [] {$c$};
\node (D4) at (9.5,1) [] {$d$};

\node (G4)     at (8.3,0) [text width=0.5cm,font=\footnotesize] {$G_4$};

\draw [->] (A4) -- (C4);
\draw [->] (C4) -- (D4);
\draw [->] (D4) -- (B4);
\draw [->] (B4) -- (A4);
\draw [->] (A4) -- (D4);
\draw [->] (C4) -- (B4);
\end{tikzpicture}
\end{center}

%
For $G_1$ and $G_2$, any subset of $X$ containing the Condorcet winner is an NE (see Proposition~\ref{prop:CondorcetWinner}). 
For $G_3$, we note that $a$ is a Condorcet loser. That is, $N(a,x)<N(x,a)$ for all $x\in\{b,c,d\}$. Note that in this case, there is no Condorcet winner in the reduced profile $P^{\downarrow \{b,c,d\}}$ as this would imply the existence of a Condorcet winner in $P$ (case $G_1$ or $G_2$). W.l.o.g., assume that $b$ beats $c$, $c$ beats $d$, and $d$ beats $b$. W.l.o.g. again, assume that $bcd \mapsto b$. Then, $\{b,c\}$ is an NE. Indeed, in any set of just two candidates, none has an incentive to leave. Now, $a$ or $d$ have no incentive to join as this would not change the winner: in the former case, observe that $b$ is the (unique) Condorcet winner in $P^{\downarrow \{a,b,c\}}$, and the latter follows by our assumption. There is always an NE for $G_3$. 

The proof for $G_4$ is more complex and proceeds case by case. Since $r$ is Condorcet-consistent, we have $acd \mapsto a$, $bcd \mapsto c$, $ab \mapsto b$, $ac \mapsto a$, $ad \mapsto a$, $bc \mapsto c$, $bd \mapsto d$ and $cd \mapsto c$. The sets of candidates for which $r$ is undetermined are $abcd$, $abc$ and $abd$.  

We have the following easy facts: {\em (i)} if $abcd \mapsto a$ then $acd$ is an NE, {\em (ii)} if $abcd \mapsto c$ then $bcd$ is an NE, 
{\em (iii)} if $abc \mapsto a$ then $ac$ is an NE, 
{\em (iv)} if $abd \mapsto a$ then $ad$ is an NE, 
{\em (v)} if $abc \mapsto c$ then $bc$ is an NE. 
The only remaining cases are: 

\begin{enumerate}
\item $abcd \mapsto b$,  $abc \mapsto b$,  $abd \mapsto b$.  
\item $abcd \mapsto b$,  $abc \mapsto b$,  $abd \mapsto d$.  
\item $abcd \mapsto d$,  $abc \mapsto b$,  $abd \mapsto b$.  
\item $abcd \mapsto d$,  $abc \mapsto b$,  $abd \mapsto d$.    
\end{enumerate}

In cases 1 and 3, $ab$ is an NE. 
In case 2, if $a$ prefers $b$ to $c$ then $abc$ is an NE, and if $a$ prefers $c$ to $b$, then $bcd$ is an NE.
In case 4, if $a$ prefers $c$ to $d$, then $bcd$ is an NE; if $b$ prefers $a$ to $d$, then $ad$ is an NE; finally, if $a$ prefers $d$ to $c$ and $b$ prefers $d$ to $a$, then $abcd$ is an NE.  
To conclude, observe that the proof never uses the fact that two profiles having the same majority graph have the same winner.\footnote{For instance, we may have two profiles $P$, $P'$ both corresponding to $G_4$, such that $r(P)=a$ and $r(P') = b$; the proof perfectly works in such a case.} \qed 
\end{proof}

Thus, the picture for four candidates shows a sharp contrast. On one hand, we show that ``almost all scoring rules''  \cite{SaariJET1989}, single transferable vote, and plurality with run-off, may fail to have an  NE. On the other hand, Condorcet-consistency alone suffices to guarantee the existence of an NE. 



\section{The second frontier: five candidates}\label{sec:5candidates}

We start with scoring rules. Recall that for four candidates we had the non-existence results for most rules, with Borda being a noticeable exception. We now show that five candidates is enough for Borda to lose this guarantee of the existence of NE. 

\begin{proposition}\label{prop:borda-5cand}
For the Borda rule, with five candidates, there are candidacy games without Nash equilibria.
\end{proposition}

\begin{proof}
The following counterexample has been obtained by applying the same ILP technique as described in the previous section. 
We do not give the profile but only its majority margin matrix, where the number corresponding to row $x$ and column $y$ is $N_P(x,y) - N_P(y,x)$; by Debord's theorem \cite{Debo87a}, the existence of a profile $P$ realizing this matrix is guaranteed because all elements of the matrix have the same parity. 

 \[
\begin{array}{|l|lllll}
\hline
& a & b & c & d & e \\
\hline
a & - & -3 & -1 & +1 & +3 \\
b & +3 & - & -5 & +1 & -1 \\
c & +1 & +5 & - & -5 & -1 \\
d & -1 & -1 & +5 & - & -3 \\
e & -3 & +1 & +1 & +3 & - \\
\hline
\end{array}
\begin{array}{|ccccc|}
\hline
a & b & c & d & e \\
\hline
a & b & c & d & e\\
b & a & a & c & c\\
e & e & d & e & d \\
c & c & e & a & a\\
d & d & b & b & b \\
\hline
\end{array}
\]

Below we give the explicit listing of all 31 states, introducing a notation that we shall use throughout the paper: the outcome of the choice function (the winner in each state) is given in boldface, and a deviation from this state is given next to each state, where $x+$ (respectively $x-$) means that $x$ has a profitable deviation by joining (respectively, by leaving) this state. It can be seen that none of the 31 states is an NE. 

$$\begin{array}{ccccc}
\begin{array}{cc}
{\mathbf a} & b+ \\
{\mathbf b} & c+ \\
{\mathbf c} & d+ \\
{\mathbf d} & e+ \\
{\mathbf e} & a+ 
\end{array}
&
\begin{array}{cc}
a {\mathbf b} & c+ \\
a {\mathbf c} & d+ \\
{\mathbf a} d & b+ \\
{\mathbf a}e & b+ \\
b{\mathbf c} & d+ \\
{\mathbf b}d & e+\\
b {\mathbf e} & c+\\
c{\mathbf d} & e+\\
c{\mathbf e} & a+\\
d{\mathbf e} & a+\\
\end{array}
&
\begin{array}{cc}
a b{\mathbf c} & d+ \\
a{\mathbf b}d & c+ \\
a{\mathbf b}e& c+ \\
ac{\mathbf d} & c- \\
{\mathbf a}ce & e- \\
{\mathbf a}de & b+ \\
bc{\mathbf d} & e+\\
b {\mathbf c}e & b-\\
bd{\mathbf e} & a+\\
cd{\mathbf e} & a+\\
\end{array}
&
\begin{array}{cc}
abc {\mathbf d} & e+ \\
ab{\mathbf c}e & b- \\
a{\mathbf b}de& c+ \\
{\mathbf a}cde & e- \\
bcd{\mathbf e} & d- \\
\end{array}
&
\begin{array}{cc}
abcd{\mathbf e} & b- 
\end{array}
\end{array} 
$$ \qed
\end{proof}

Recall that for Condorcet-consistent rules, the existence of NE is guaranteed for four candidates. For the maximin rule and the uncovered set rule, this existence result stops at four. The proof, consisting of two counterexamples, is in Appendix A.

\begin{proposition}\label{prop:maximin-5cand}
For the maximin rule and the uncovered set rule, with five candidates, there are candidacy games without  NE.
\end{proposition}


However, this negative result does not extend to all Condorcet-consistent rules, as shown in Proposition~\ref{prop:topcycle-5cand} below (and also in Proposition~\ref{prop:Copeland-morecand} in the following section).

\begin{proposition}\label{prop:topcycle-5cand}
For the Top-Cycle rule, with five candidates, every candidacy game has a Nash equilibrium. 
\end{proposition}

\begin{proof}
Let $P$ be a profile over $X = \{a,b,c,d,e\}$ and without loss of generality, assume that the tie-breaking priority ranks $a$ above all other candidates. 
If $|TC(P)| \leq 4$ then consider the restriction $P^{\downarrow TC(P)}$ of $P$ to $TC(P)$. It is a $q$-candidate profile for $q \leq 4$, therefore by Proposition \ref{Condorcet-four} the corresponding candidacy game has an NE $Z \subseteq TC(P)$. Because it is an NE in $P^{\downarrow TC(P)}$, no candidate in $TC(P)$ has an incentive to deviate. Now, if a candidate in $X \setminus TC(P)$ joins, the outcome does not change, therefore no candidate outside $TC(P)$ has an incentive to join. Therefore, $Z$ is an NE for $P$.

Assume now that $TC(P) = \{a,b,c,d,e\}$; this implies $TC_t(P) = a$. Without loss of generality, assume 
the majority graph contains $a \to b \to c \to d \to e \to a$. 
For $abcde$ not to be an NE, a withdrawing agent $x$ has to induce a new top-cycle \emph{not} containing $a$. If this top-cycle is a singleton, then 
$X \setminus \{x\}$ is an NE. Therefore, the top-cycle after the withdrawal of $x$ must be of size 3: it can only be $\{c,d,e\}$, with $b$ withdrawing because it prefers 
the most prioritary candidate (let us call it $y$)  among $\{c,d,e\}$  
to $a$. At this stage, we know that $d \to a$, $c \to a$, $e \to a$, $c \to d \to e \to c$, and that the winner in $acde$ is $y$. 
Observe that, irrespective of the tie-breaking winner, $a$ cannot leave because the winner would remain the same. There are thus three cases to consider: 
\begin{itemize}
\item {\em Case 1: $y = c$}. Consider $acd \mapsto c$. Since $ac \mapsto c$, $cd \mapsto c$, and $acde \mapsto c$, $acd$ is not an NE only if $b$ wants to join; but $abcd \mapsto a$, and $b$ prefers $c$ to $a$:  $bcd$ is an NE. 
\item {\em Case 2: $y = e$}. Consider $ace \mapsto e$. Since $ae \mapsto e$, $ce \mapsto e$, and $acde \mapsto e$, $ace$ is not an NE only if $b$ wants to join. For this to be possible, we must have $b \to e$, and then $abce \mapsto a$. But in this case, since $abc \mapsto a$, $abe \mapsto a$, and $abcde \mapsto a$,  $abce$ is an NE. Therefore, either $ace$ or $abce$ is an NE.  
\item {\em Case 3: $y = d$}. Consider $ade \mapsto d$. Since $ad \mapsto d$, $de \mapsto d$ and $acde \mapsto d$, $ade$ is not an NE only if $b$ wants to join. For this to be possible, it must be that $b \to d$ (and $b$ prefers $a$ over $d$). Thus $abde \mapsto a$. In this case, since $abd \mapsto a$ and $abcde \mapsto a$, $abde$ is not an NE only if  $d$ wants to leave. This is possible only if $e \to b$ (and $d$ prefers $e$ over $a$). But then $abe \mapsto e$, $ae \mapsto e$, $be \mapsto e$, and $abce \mapsto e$: $abde$ is an NE.  Therefore, either $ade$ or $abde$ is an NE. \qed
\end{itemize}
\end{proof}



\section{More candidates}\label{sec:more}
In this section, we present our results for a general number of candidates.
\subsection{A positive result: Copeland}

We show the existence of NE for Copeland, under deterministic tie-breaking, for any number of candidates (provided $n$ is odd).
%
\begin{proposition}\label{prop:Copeland-morecand}
For Copeland, for any number of candidates and an odd number of voters, every candidacy game has an NE.
\end{proposition}

\begin{proof}
Let $P$ be a profile and $\rightarrow_P$ its associated majority graph. Let $C(x,P)$ be the number of candidates $y \neq x$ such that $x \rightarrow_P y$. The Copeland cowinners for $P$ are the candidates maximizing $C(\cdot,P)$. 

Let $Cop(P)$ be the set of Copeland cowinners for $P$ and let $c$ be the Copeland winner---i.e., the most prioritary candidate in $Cop(P)$. Consider $Dom(c) = \{c\} \cup \{y | c \rightarrow_P y\}$. Note that $C\left(c,P^{\downarrow Dom(c)}\right)=|Dom(c)| -1 = q \geq C(c,P)$. Also, since any $y \in Dom(c)$ is beaten by $c$, we have  $C(y,P^{\downarrow Dom(c)} )\leq q-1$.

We claim that $Dom(c)$ is an NE. Note that  $c$ is a Condorcet winner in the restriction of $P$ to $Dom(c)$, and {\em a fortiori}, in the restriction of $P$ to any subset of $Dom(c)$. Hence, $c$ is the Copeland winner in $Dom(c)$ and any of its subsets, and no candidate in $Dom(c)$ has an incentive to leave. 

Now, assume there is a candidate $z \in X \setminus Dom(c)$ such that  $r\left(P^{\downarrow Dom(c) \cup \{z\}}\right)\neq c$. Note that $z \rightarrow_P c$ as $z$ does not belong to $Dom(c)$; so, $C(c,P^{\downarrow Dom(c) \cup \{z\}} ) = q$. 

For any $y \in Dom(c)$ we have $C(y,P^{\downarrow Dom(c) \cup \{z\}} )\leq (q-1)+1=q= C(c,P^{\downarrow Dom(c) \cup \{z\}} )$. If $C\left(y,P^{\downarrow Dom(c) \cup \{z\}} \right)< C(c,P^{\downarrow Dom(c) \cup \{z\}} )$, then $y$ is not the Copeland winner in $P^{\downarrow Dom(c) \cup \{z\}}$. If $C\left(y,P^{\downarrow Dom(c) \cup \{z\}} \right) = C\left(c,P^{\downarrow Dom(c) \cup \{z\}}\right)$, then $C(y,P)\geq C(c,P)$. That is, either $c\notin Cop(P)$, a contradiction, or both $y,c$ are in $Cop(P)$. The latter implies $c \vartriangleright y$; hence, $y$ is not the Copeland winner in $P^{\downarrow Dom(c) \cup \{z\}}$. 

Hence, $r\left(P^{\downarrow Dom(c) \cup \{z\}}\right) = z$. That is, either (1) $C\left(z,P^{\downarrow Dom(c) \cup \{z\}}\right)> q$, or (2) $C\left(z,P^{\downarrow Dom(c) \cup \{z\}}\right) = q$ and $z \vartriangleright c$. If (1) holds then $C(z,P) > C(c,P)$, which contradicts the fact that $c$ is the Copeland winner in $P$. If (2) holds then $C(z,P) = C(c,P)$---i.e.,  both $c$ and $z$ are in $Cop(P)$, which implies that $c \vartriangleright z$, and $z$ cannot win in $P^{\downarrow Dom(c) \cup \{z\}}$. Therefore, the Copeland winner in $P^{\downarrow Dom(c) \cup \{z\}}$ is $c$, which implies that $z$ has no incentive to join $Dom(c)$. \qed
\end{proof}

Note that not only the existence of an NE is guaranteed, but also the existence of an NE {\em where the winner is the same winner as on the original profile} (that is, the Copeland winner of the profile with all candidates running).\footnote{Note however that this does {\em not} imply that the set of all candidates is an NE. For instance, let $X = \{a,b,c,d\}$, and consider the majority graph $a \to b,a \to c; b \to c, b \to d;  d \to a, d\to c$, with the tie-breaking priority relation $a \tb b \tb c \tb d$. The Copeland winner is $a$ (by tie-breaking). We only need to specify that $b: d \succ a$ on top of self-supported preferences. $X$ is not an NE, because it is a profitable deviation for $b$ to leave.}

When $n$ is even, the result carries on if no pairwise majority ties occur. In the general case, however, the result depends on the way ties are taken into account for computing the Copeland score of a candidate. For the variant Copeland$^0$ where the Copeland score remains the number of outgoing edges (ties not giving any point), the result still holds. Whether it holds for other variants is an open question.


\subsection{Top Cycle} 


\begin{proposition}\label{prop:topcycle-7cand}
For the Top-Cycle rule, with six candidates, every candidacy game has an NE, and with seven candidates, there are candidacy games without an NE. 
\end{proposition}

Both results have been obtained by computer search. 
Technically, we first pruned the domain to reduce the number of majority graphs to consider. Then, for each remaining graph, we computed the \emph{co-winners} given by the top-cycle rule, and we launched a feasibility problem asking the computer to build an instance without equilibrium. This is similar in spirit to the ones used in previous sections, but including additional decision variables for the tie-breaking ordering (and making sure that winners are indeed among the co-winners). 
For the six candidate case, the infeasibility of the program tells us that an equilibrium always exists, but we could not extract any 
readable proof from the result.\footnote{Note that this positive result holds as well for the Banks rule, since Top-Cycle and Banks do coincide up to six candidates \cite{BDS2015}.} 
The counterexample for seven candidates is given in Appendix A.

\subsection{More negative results by induction}

For all rules that satisfy IBC and for which we have already found a counter-example for $m$, we know that counterexamples exist for any number of candidates. As we previously noted, veto is an example of a rule not satisfying IBC, but an adapted version of Lemma  \ref{lemma-more-candidates} can easily be designed (see Lemma \ref{lemma-more-candidates-veto} in Appendix).  
As a corollary of these, and of Propositions \ref{saari4}, \ref{pro-stv4}, \ref{prop:maximin-5cand}, \ref{prop:borda-5cand} and \ref{prop:topcycle-7cand} we get:

\begin{corollary}~ There exists profiles with no NE in the following cases:
\begin{itemize}
\item For all Saari scoring rules satisfying IBC (including plurality),
 as well as for veto, for all $m \geq 4$.
\item For plurality with runoff and single transferable vote, for all $m \geq 4$.
\item For Borda, maximin, and the uncovered set, for all $m \geq 5$.
\item For TopCycle, and for all $m \geq 7$.
\end{itemize}
\end{corollary}



\section{Strong Equilibria and Link to Control}\label{sec:control}

\subsection{Strong Nash Equilibria}

We now prove that the lack of guarantee for the existence of strong Nash equlilibria holds for almost any voting rule and any number of candidates $m \geq 3$. 

Let $r$ be a voting rule defined for a varying set of candidates $Y \subseteq X$. 
We say that $r$  is {\em majority-extending} if for any $Y \subseteq X$ such that $|Y| = 2$ and if the two candidates in $Y$ are not tied in $P^{\downarrow Y}$ then $r(P^{\downarrow Y})$ is the majority winner in $P^{\downarrow Y}$ (in case of a tie, we don't need to specify the outcome).


\begin{proposition}
There does not exist any majority-extending and IBC rule  that guarantees the existence of an SE at every profile. 
\end{proposition}

\begin{proof}
Let $r$ be a majority-extending and IBC rule.
Consider the following 3-voter, $k+3$-candidate profile $(k \geq 0$):

\begin{footnotesize}
\[
\begin{array}{|ccc|}
\hline
1 & 1 & 1 \\
\hline
a & b & c \\
b & c &  a \\
c & a &  b \\
x_1 & x_1 & x_1 \\
\vdots & \vdots & \vdots \\
x_k & x_k & x_k \\
\hline
\end{array} 
\]
\end{footnotesize}


By a repeated application of IBC, for any nonempty $Y \subseteq \{a,b,c\}$ and any $Z \subseteq \{x_1, \ldots, x_k\}$ we have $r(P^{\downarrow Y \cup Z}) = r(P^Y)$. 

We already know that $r(P^{\downarrow \{a,b,c, x_1, \ldots, x_k\}} \in \{a,b,c\}$. 
Without loss of generality, assume that $r(P^{\downarrow \{a,b,c, x_1, \ldots, x_k\}}) = a$. For any $Z \subseteq \{x_1, \ldots, x_k\}$, by IBC and majority-extension, the resulting choice function must be: 

$$abcZ \mapsto a ; abZ \mapsto a ; bcZ \mapsto b ; acZ \mapsto c ; aZ \mapsto a ; bZ \mapsto b ; cZ \mapsto c $$



But then, given the candidates' preferences, for any $Z \subseteq \{x_1, \ldots, x_k\}$ we have: 

\begin{itemize}
\item  $abcZ$ is not an SE: $abcZ \mapsto a$, $b$ leaves $\mapsto c$
\item  $abZ$ is not an SE: $abZ \mapsto a$, $b$ leaves and $c$ joins $\mapsto c$
\item  $acZ$ is not an SE: $acZ \mapsto c$, $a$ leaves and $b$ joins $\mapsto b$
\item  $bcZ$ is not an SE: $bcZ \mapsto b$, $a$ joins $\mapsto a$
\item  $aZ$ is not an SE: $aZ \mapsto a$, $c$ joins $\mapsto c$
\item  $bZ$ is not an SE: $bZ \mapsto b$, $a$ joins $\mapsto a$
\item  $cZ$ is not an SE: $cZ \mapsto c$, $b$ joins $\mapsto b$
\item  $Z$ is not an SE: any of $a$, $b$ or $c$ wants to join. \qed
\end{itemize}
\end{proof}

The result applies to most common voting rules.\footnote{A noticeable exception is veto;
however, we already know that for veto,
there exist profiles without NE, and therefore without SE.}

\subsection{Relation to Control}


Bartholdi et al.~\cite{BartholdiToveyTrick92} define
 \emph{constructive control by deleting candidates} (CCDC) and \emph{constructive control by adding candidates} (CCAC): 
an instance of CCDC 
consists of a profile $P$ over set of candidates $C$, a distinguished candidate $c$, an integer $k$, and we ask whether there is a subset $C'$ of $C$ with $|C \setminus C'| \leq k$ such that $c$ is the unique winner in $C'$. An instance of CCAC consists of a profile $P$ over set of candidates $C_1 \cup C_2$, a distinguished candidate $c$, and we ask whether there is a subset $C'$ of $C_2$ such that the unique winner in $C_1 \cup C'$ is $c$.
{\em Destructive} versions of control are defined by Hemaspaandra et al.~\cite{HemaspaandraHR07}: \emph{destructive control by deleting} (DCDC) is similar to CCDC, except that we ask whether there is a subset $C'$ of $C \setminus \{c\}$ with $|C \setminus C'| \leq k$ such that $c$ is {\em not} the unique winner in $C \setminus C'$; and  
\emph{destructive control by adding candidates} (DCAC) is similar to CCAC, except that $c$ should {\em not} be the unique winner in $C'$.  There are also {\em multimode} versions of control~\cite{FHH11}: e.g., CC(DC+AC) allows the chair to delete some candidates {\em and} to add some others (subject to some cardinality constraints).

Nash equilibria and strong equilibria in strategic candidacy relate to a slightly more demanding notion of control, which we can call {\em consenting control}, and that we find an interesting notion {\em per se}. In traditional control, candidates have no preferences and no choice---the chair may add or delete them as he likes. An instance of {\em consenting CCDC} consists of an instance of CCDC plus, for each candidate in $C$, a preference ranking  over $C$, and we ask whether there is a subset $C'$ of $C$ with $|C \setminus C'| \leq k$ such that $c$ is the unique winner in $C'$, and every candidate in $C \setminus C'$ prefers $c$ to the candidate which would win if all candidates in $C$ were running. An instance of {\em consenting CCAC} consists of an instance of CCAC plus, for each candidate in $C_2$, a ranking over $C_1 \cup C_2$, and we ask whether there is a subset $C'$ of $C_2$ such that $c$ is the unique winner in $C_1 \cup C'$  and every candidate in $C'$ prefers $c$ to the candidate which would win if only the candidates in $C_1$ were running. Consenting versions of destructive control 
are defined similarly: here the goal is to have a different candidate from the current winner elected. 

Clearly, for profile $P$,  $(1,\ldots,1)$ is an SE iff there is no consenting destructive control by removing candidates against the current winner $r(X)$, with the value of $k$ being fixed to $m$ (the chair has no limit on the number of candidates to be deleted; the limits come here from the fact that the candidates must consent), and $(1,\ldots,1)$ is an NE iff there  is no consenting destructive control by removing candidates against the current winner $r(X)$, with the upper bound of $k = 1$ on the number of candidates to be deleted. 

For candidate sets that are different from the set $X$ of all candidates (as some may leave and some other may join), we have to resort to consenting destructive control by removing {\em and}  adding candidates, as in \cite{FHH11}. Let $s$ be a state and $X_s$ the set of running candidate in $s$:  $s$ is an SE if there is no consenting destructive control by removing and adding candidates against the current winner $r(X_s)$, without any constraint on the number of candidates to be removed or added. For an NE, this is similar, but with the bound $k = 1$ on the number of candidates to be deleted or added.


\section{Conclusions}\label{sec:further}

We have explored further the landscape of strategic candidacy in elections by obtaining several positive results
 and several negative results 
which can be summarized on the following table, where ``yes$^*$'' means yes under the assumption that $n$ is odd, or more generally that pairwise ties do not occur.

\begin{footnotesize}
\begin{center}
\begin{tabular}{c|c|c|c|c}
& 3 & 4 & 5-6 & $\geq$ 7\\ \hline
plurality & yes & no & no & no \\
veto & yes & no & no & no \\
pl. runoff & yes & no & no & no \\
STV & yes & no & no & no \\
Borda & yes & yes & no & no \\
maximin & yes & yes$^*$ & no & no \\
UC & yes & yes$^*$ & no & no \\
TC & yes & yes$^*$ & yes$^*$ & no \\
Copeland& yes & yes$^*$ & yes$^*$ & yes$^*$ \\
\hline
\end{tabular}
\end{center}
\end{footnotesize}

An important issue for further research is a characterization of all rules for which the existence of a pure Nash equilibrium is guaranteed, at least for an odd number of voters. We know that not only it contains Copeland, as well as the rule defiend by the sophisticated winner of the successive elimination rule; these two rules do not have much in common, which suggests that such a characterization could be highly complex.

Another issue is the study of the set of states that can be reached by some (\emph{e.g.} best response) dynamics starting from the set or all potential candidates. In some cases, even when the existence of NE is guaranteed (\emph{e.g.} for Copeland), we could already come up with examples such that none is reachable by a sequence of best responses. But other types of dynamics may be studied. Another issue for further research is the computational complexity of deciding whether there is an NE or SE. 

Finally, a recent line of research, dealing with a setting where not only candidates, but also voters, are strategic players, has been investigated by Brill and Conitzer \cite{BrillConitzerAAAI2015}. 


\section*{Acknowledgements} We would like to thank Markus Brill, Edith Elkind, Michel Le Breton and Vincent Merlin for helpful discussions.

{\small 
\bibliographystyle{abbrv}

}




\section*{Appendix A}

\noindent {\bf Proposition \ref{prop:CondorcetWinner}}
{\em Let $\Gamma=\langle X, P, r, P^X\rangle$ be a candidacy game where $r$ is Condorcet-consistent. If $P$ has a Condorcet winner $c$ then for any $Y \subseteq X$, 
\begin{center}$Y$ is a SE $\Leftrightarrow$ $Y$ is an NE $\Leftrightarrow$ $c \in Y$.\end{center}}
\begin{proof}
Assume $c$ is a Condorcet winner for $P$ and let $Y \subseteq X$ such that $c \in Y$. Because $r$ is Condorcet-consistent, and because $c$ is a Condorcet winner for $P^{\downarrow Y}$, we have $r\left(P^{\downarrow Y}\right) = c$. Assume $Z=Z^+\cup Z^-$ is a deviating coalition from $Y$, with $Z^+$ the candidates who join and $Z^-$ the candidates who leave the election. Clearly, $c \notin Z$, as $c \in Y$ and $c$ has no interest to leave. Therefore, $c$ is still a Condorcet winner in $P^{\downarrow (Y \setminus Z^-) \cup Z^+}$, which by the Condorcet-consistency of $r$ implies that $r\left( P^{\downarrow(Y \setminus Z^-) \cup Z^+}\right) = c$, which contradicts the assumption that $Z$ wants to deviate. We thus conclude that $Y$ is an SE, and {\em a fortiori} an NE.
Finally, let $Y \subseteq X$ such that $c \notin Y$. Then, $Y$ is not an NE (and {\em a fortiori} not an SE), because $c$ has an interest to join the election.
\qed
\end{proof}

\noindent {\bf Proposition \ref{prop:maximin-5cand}}
{\em For maximin and the uncovered set, with five candidates, there are profiles with no NE.}
\begin{proof}
For maximin, a counterexample is the following weighted majority graph along with the candidates' preference profile. 
The tie-breaking priority is lexicographic. 

\[
\begin{array}{|l|lllll}
\hline
& a & b & c & d & e \\
\hline
a & - & 1 & 4 & 2 & 3 \\
b & 4 & - & 1 & 4 & 3 \\
c & 1 & 4 & - & 2 & 2 \\
d & 3 & 1 & 3 & - & 0 \\
e & 2 & 2 & 3 & 5 & - \\
\hline
\end{array}
\begin{array}{|ccccc|}
\hline
a & b & c & d & e \\
\hline
a & b & c & d & e\\
c & e & d & a & b\\
b & c & a & c & a \\
e & a & e & b & d \\
d & d & b & e & c \\
\hline
\end{array}
\]
\noindent

Below we give all 31 states, with the usual notation.

$$\begin{array}{ccccc}
\begin{array}{cc}
{\mathbf a} & b+ \\
{\mathbf b} & c+ \\
{\mathbf c} & a+ \\
{\mathbf d} & b+ \\
{\mathbf e} & a+ 
\end{array}
&
\begin{array}{cc}
a {\mathbf b} & c+ \\
{\mathbf a}  c & d+ \\
a {\mathbf d} & b+ \\
{\mathbf a} e & b+ \\
b {\mathbf c} & a+ \\
{\mathbf b} d & c+\\
{\mathbf b} e & c+\\
c {\mathbf d} & b+\\
c{\mathbf e} & a+\\
d {\mathbf e} & a+\\
\end{array}
&
\begin{array}{cc}
{\mathbf a} bc & e+ \\
a{\mathbf b}d & c+ \\
a{\mathbf b}e& c+ \\
ac{\mathbf d} & b+ \\
{\mathbf a}ce & b+ \\
{\mathbf a}de & b+ \\
b{\mathbf c}d & a+\\
b {\mathbf c}e & b-\\
{\mathbf b}de & c+\\
cd{\mathbf e} & a+\\
\end{array}
&
\begin{array}{cc}
{\mathbf a} bcd & e+ \\
abc{\mathbf e} & a- \\
a{\mathbf b} de& c+ \\
{\mathbf a}cde & b+ \\
b{\mathbf c}de & b- \\
\end{array}
&
\begin{array}{cc}
abcd{\mathbf e} & a- 
\end{array}
\end{array}
$$

Here is now a counter-example for the uncovered set rule. The tie-breaking rule is
$a \tb  b \tb  d \tb  c \tb  e$.

\begin{minipage}{5cm}
\[
\begin{array}{|l|lllll}
\hline
& a & b & c & d & e \\
\hline
a & 0 & 0 & 0 & 1 & 0 \\
b & 1 & 0 & 1 & 0 & 0 \\
c & 1 & 0 & 0 & 1 & 0 \\
d & 0 & 1 & 0 & 0 & 1 \\
e & 1 & 1 & 1 & 0 & 0 \\
\hline
\end{array}
\begin{array}{|ccccc|}
\hline
a & b & c & d & e \\
\hline
a & b & c & d & e\\
e & e & e & b & b\\
c & c & d & a & a \\
b & a & a & e & c \\
d & d & b & c & d \\
\hline
\end{array}
\]
\end{minipage}
\begin{minipage}{6cm}
$$\begin{array}{ccccc}
\begin{array}{cc}
{\mathbf a} & b+ \\
{\mathbf b} & d+ \\
{\mathbf c} & b+ \\
{\mathbf d} & a+ \\
{\mathbf e} & d+ 
\end{array}
&
\begin{array}{cc}
a {\mathbf b} & e+ \\
a {\mathbf c} & b+ \\
{\mathbf a} d & c+ \\
a {\mathbf e} & d+ \\
{\mathbf b}c & e+ \\
b {\mathbf d} & a+\\
b {\mathbf e} & d+\\
{\mathbf c}d & b+\\
c{\mathbf e} & d+\\
{\mathbf d}e & a+\\
\end{array}
&
\begin{array}{cc}
a {\mathbf b}c & e+ \\
{\mathbf a}bd & d- \\
ab{\mathbf e}& d+ \\
a{\mathbf c}d & b+ \\
ac {\mathbf e} & d+ \\
{\mathbf a}de & c+ \\
{\mathbf b}cd & c-\\
b c{\mathbf e} & d+\\
b{\mathbf d}e & a+\\
c{\mathbf d}e & e-\\
\end{array}
&
\begin{array}{cc}
a {\mathbf b}cd & a- \\
abc{\mathbf e} & d+ \\
{\mathbf a}bde& c+ \\
ac{\mathbf d}e & e- \\
bc{\mathbf d}e & e- \\
\end{array}
&
\begin{array}{cc}
abc{\mathbf d}e & e- 
\end{array}
\end{array}
$$
\end{minipage}

\qed
\end{proof}

\noindent {\bf Proposition \ref{prop:topcycle-7cand}}
{\em For the Top-Cycle rule and seven candidates, there are profiles with no NE.}

\begin{proof}
We give the majority graph, 
tie-breaking relation, and the (partially specified) candidates' preferences. 
The tie-breaking relation is $a \tb g \tb c \tb b \tb d \tb e \tb f$.


\begin{minipage}{5cm}
\begin{tikzpicture}[scale=0.5]
\tikzstyle{is-in}=[circle,draw, line width=1pt]
\tikzstyle{is-out}=[circle,draw, line width=1pt,color=gray]

\node (A) at (0,3) [is-in,label=above:\footnotesize{$a$}, label=below:\footnotesize{}] {};
\node (B) at (3,6) [is-in,label=above:\footnotesize{$b$}, label=below:\footnotesize{}] {};
\node (C) at (6,6) [is-in,label=above:\footnotesize{$c$}, label=below:\footnotesize{}] {};
\node (D) at (9,4.5) [is-in,label=above:\footnotesize{$d$}, label=above:\footnotesize{}] {};
\node (E) at (9,1.5) [is-in,label=below:\footnotesize{$e$}, label=below:\footnotesize{}] {};
\node (F) at (6,0) [is-in,label=below:\footnotesize{$f$}, label=below:\footnotesize{}] {};
\node (G) at (3,0) [is-in,label=below:\footnotesize{$g$}, label=below:\footnotesize{}] {};


\draw [->] (A) -- (B);
\draw [->] (B) -- (C);
\draw [->] (C) -- (D);
\draw [->] (D) -- (E);
\draw [->] (E) -- (F);
\draw [->] (F) -- (G);
\draw [->] (G) -- (A);

\draw [->] (B) -- (D);
\draw [->] (B) -- (G);
\draw [->] (C) -- (A);
\draw [->] (C) -- (E);
\draw [->] (D) -- (A);
\draw [->] (D) -- (G);
\draw [->] (E) -- (A);
\draw [->] (E) -- (B);
\draw [->] (E) -- (G);
\draw [->] (F) -- (A);
\draw [->] (F) -- (B);
\draw [->] (F) -- (C);
\draw [->] (F) -- (D);
\draw [->] (F) -- (G);
\draw [->] (G) -- (C);

\end{tikzpicture}
\end{minipage}
~
%
%
%
\begin{minipage}{5cm}
\[
\begin{array}{|ccccccc|}
\hline
a & b & c & d & e & f & g \\
\hline
a & b & c & d & e & f & g \\
 & c & d & c & g &  &  \\
 & f & b & f & d &  & \\
 & g & a & g & a &  & \\
 & a & f & b & b  &  & \\
 & d & e & e & c &  & \\
 & e & g & a & f &  & \\

\hline
\end{array}
\]
\end{minipage}


\small
$$\begin{array}{ccccccc}
\begin{array}{cc}
{\mathbf a} & g+ \\
{\mathbf b} & a+ \\
{\mathbf c} & f+ \\
{\mathbf d} & c+ \\
{\mathbf e} & c+ \\
{\mathbf f} & e+ \\
{\mathbf g} & e+ 

\end{array}
&
\begin{array}{cc}
{\mathbf a}b & f+ \\
a {\mathbf c} & f+ \\
a{\mathbf d} & c+ \\
a {\mathbf e} & d+ \\
a {\mathbf f} & e+ \\
a {\mathbf g} & e+ \\
{\mathbf b}c & a+ \\
{\mathbf b}d & a+\\
b {\mathbf e} & c+\\
b {\mathbf f} & e+\\
{\mathbf b}g & f+\\
{\mathbf c}d & g+\\
{\mathbf c}e & g+\\
{\mathbf c}f & e+\\
c{\mathbf g} & f+\\
{\mathbf d}e & c+\\
d{\mathbf f} & e+\\
{\mathbf d}g & f+\\
{\mathbf e}f & c+\\
{\mathbf e}g & d+\\
{\mathbf f}g & e+\\
\end{array}
& 
\begin{array}{cc}
{\mathbf a}bc & f+ \\
{\mathbf a}bd & f+ \\
ab{\mathbf e} & c+ \\
ab{\mathbf f} & e+ \\
{\mathbf a}bg & f+ \\
a{\mathbf c}d & f+ \\
a{\mathbf c}  e & g+ \\
ac{\mathbf f} & e+ \\
ac{\mathbf g} & f+ \\
a{\mathbf d}e & c+ \\
ad{\mathbf f} & e+ \\
a{\mathbf d}g & f+ \\
a{\mathbf e}f & d+ \\
a{\mathbf e}g & d+ \\
a{\mathbf f}g & e+ \\
{\mathbf b}cd & f+\\
b{\mathbf c}e & g+\\
bc{\mathbf f} & e+\\
{\mathbf b}cg & f+\\
{\mathbf b}de & c+\\
bd{\mathbf f} & e+\\
{\mathbf b}dg & a+\\
b{\mathbf e}f & c+\\
b{\mathbf e}g & d+\\
b{\mathbf f}g & e+\\
{\mathbf c}de & g+\\
cd{\mathbf f} & e+\\
cd{\mathbf g} & c-\\
{\mathbf c}ef & g+\\
ce{\mathbf g} & c-\\
c{\mathbf f}g & d+\\
{\mathbf d}ef & b+\\
{\mathbf d}eg & b+\\
 d{\mathbf f}g & e+ \\
{\mathbf e}fg & d+\\
\end{array}
& 
\begin{array}{cc}
{\mathbf a}bcd & f+ \\
{\mathbf a}bce & b- \\
abc{\mathbf f} & e+ \\
{\mathbf a}bcg & f+ \\
{\mathbf a}bde & d- \\
abd{\mathbf f} & e+ \\
{\mathbf a}bdg & f+ \\
ab{\mathbf e}f & c+ \\
ab{\mathbf e}g & c+ \\
ab{\mathbf f}g & e+ \\
a {\mathbf c}de & g+ \\
acd{\mathbf f} & e+ \\
acd{\mathbf g} & c- \\
a {\mathbf c}ef & g+ \\
ace{\mathbf g} & c- \\
ac{\mathbf f}g & e+ \\
a {\mathbf d}ef & c+ \\
a {\mathbf d}eg & b+ \\
ad{\mathbf f}g & e+ \\
a {\mathbf e}fg & d+ \\

b {\mathbf c}de & e- \\
bcd{\mathbf f} & e+ \\
{\mathbf b}cdg & a+ \\
b {\mathbf c}ef & g+ \\
bce{\mathbf g} & a+ \\
bc{\mathbf f}g & e+ \\
{\mathbf b}def & a+ \\
{\mathbf b}deg & a+ \\
bd{\mathbf f}g & e+ \\
b {\mathbf e}fg & d+ \\

{\mathbf c}def & g+ \\
cde{\mathbf g} & c- \\
cd{\mathbf f}g & e+ \\
cef{\mathbf g} & c- \\
{\mathbf d}efg & b+ \\

\end{array}
& 
\begin{array}{cc}
{\mathbf a}bcde & b- \\ 
abcd{\mathbf f} & e+ \\
{\mathbf a}bcdg & f+ \\
{\mathbf a}bcef & b- \\
{\mathbf a}bceg & b- \\
abc{\mathbf f}g & e+ \\
{\mathbf a}bdef & d- \\
{\mathbf a}bdeg & d-  \\
abd{\mathbf f}g & e+ \\
ab{\mathbf e}fg & c+  \\

a{\mathbf c}def & g+  \\
acde{\mathbf g} & c- \\
acd{\mathbf f}g & e+ \\
acef{\mathbf g} & c- \\
a{\mathbf d}efg & b+ \\
b{\mathbf c}def & g+ \\
bcde{\mathbf g} & c-  \\
bcd{\mathbf f}g & e+ \\
bcef{\mathbf g} & c-  \\

{\mathbf b}defg & a+ \\
cdef{\mathbf g} & c- \\

\end{array}

& 
\begin{array}{cc}
{\mathbf a}bcdef & b-  \\ 
{\mathbf a}bcdeg & b- \\ 
abcd{\mathbf f}g & e+ \\ 
{\mathbf a}bcefg & b- \\ 
{\mathbf a}bdefg & d-  \\ 
 acdef{\mathbf g} & c-  \\ 
bcdef{\mathbf g} & c-  \\ 
\end{array}

& 
\begin{array}{cc}
{\mathbf a}bcdefg & b-  \\ 
\end{array}

\end{array}
$$
\normalsize

\qed
\end{proof}

\begin{lemma}\label{lemma-more-candidates-veto}
For the veto rule $r$, if there exists $\Gamma = \langle X, P, r, P^X\rangle$ with no NE, then there exists $\Gamma' = \langle X', P', r, P^{X'}\rangle$ with no NE, where $|X'| = |X| + 1$.
\end{lemma}

\begin{proof}
Take  $\Gamma$ with no NE, with $X = \{x_1, \ldots, x_m\}$, and $n$ voters $n$ odd. Let $s(P,Y,x_i)$ denote the veto score of $x_i$ in $P^{\downarrow Y}$.
Let $X' = X \cup \{x_{m+1}\}$, and $Q$ be the following $3n$-voter profile: for each vote $P_i$ in $P$ we have two identical votes $Q_i$, $Q_i'$ obtained from $P_i$ by adding $x_{m+1}$ in the bottom position, and one vote $Q_i''$ obtained from $P_i$ by adding $x_{m+1}$ in the top position.
Finally, let $P^{X'}$ be the candidate profile obtained by adding $x_{m+1}$ at the bottom of every ranking of a candidate $x_i$, $i < m$, and whatever ranking for $x_{m+1}$. Let $\Gamma' = \langle X', P', r, P^{X'}\rangle$

Let $Y \subseteq X$ and  $Y' = Y \cup \{x_{m+1}\}$.

For all $x_i \in Y$, $s(Q,Y,x_i) = 3s(P,Y\setminus \{x_{m+1} \},x_i)$; therefore, $r(Q^{\downarrow Y}) = r(P^{\downarrow Y})$. Because $Y$ is not an NE for $\Gamma$, some candidate $x_i \in X$ has a profitable deviation from $Y$ in $\Gamma$, thus $x_i$ has a profitable deviation in from $Y$ in $\Gamma'$ too: $Y$ is not an NE in $\Gamma'$.

For all $x_i \in Y$, $s(Q,Y',x_i) = s(P,Y\setminus \{x_{m+1} \},x_i) + 2n \geq 2n$, while $s(Q,Y',x_{m+1}) = n$; therefore, $r(Q^{\downarrow Y'}) = r(P^{\downarrow Y})$, and a profitable deviation from $Y$ in $\Gamma$ is also a profitable deviation in from $Y'$ in $\Gamma'$ too: $Y'$ is not an NE in $\Gamma'$.
\qed

\end{proof}

\section*{Appendix B: ILP formulation}

Let $S$ be the set of states, and $A(s)$ be the set of agents who are candidates in state $s$.  Note that $|S| = 2^{|X|}$.

\paragraph{Choice functions without any NE.} We introduce a binary variable $w_{si}$, meaning that agent $i$ wins in state $s$. 
We  add constraints enforcing that there is a single winner in each state $s$ :
\begin{eqnarray}
\forall i\in X, \forall s\in S : & w_{s,i}\in\{0,1\}
\label{eqC1}\\
\forall s\in S : & \sum_{i\in X}{w_{s,i}}=1
\label{eqC2}\\
\forall s\in S, \forall i \in X \not \in A(s) : & {w_{s,i}}=0
\label{eqC3}
\end{eqnarray}

Now, we introduce constraints related to deviations. We denote by $D(s)$ the set of \emph{possible deviations} from state $s$ (state where a single agent's candidacy differs from $s$).  We also denote by $a(s,t)$ an agent potentially deviating from $s$ to $t$. 
We define binary variables $d_{s,t}$ indicating a deviation from a state $s$ to a state $t$.  
In each state, there must be at least one deviation otherwise this state must be a NE. 
\begin{eqnarray}
\forall s\in S, \forall t\in S : & d_{s,t}\in\{0,1\} \label{eqC4}\\
\forall s\in S : &  \sum_{t\in D(s)}{d_{s,t}} \geq 1 \label{eqC5}
\end{eqnarray}

Now we introduce constraints related to the preferences of the candidates. 
For this purpose, we introduce a binary variable $p_{i,j,k}$, meaning that agent $i$ prefers candidate $j$ over candidate $k$. 
If there is indeed a deviation from $s$ to $t$, the deviating agent must prefer the winner of the state new state compared to the winner of the previous state: 
\begin{eqnarray}
\forall s\in S, \forall t\in D(s), \forall i\in X, \forall j \in X : & w_{s,i} + w_{t,j} + d_{s,t} - p_{a(s,t),j,i} \leq 2 \label{eqC6}
\end{eqnarray}
Finally we ensure that the preferences are irreflexive and transitive, and respect the constraint of being self-supported. 
\begin{eqnarray}
  \forall i\in X, \forall j \in X : & p_{i,j,j} = 0 \label{eqC7}\\
  \forall a\in X, \forall i \in X \forall j\in X, \forall k \in X : & p_{a,i,j} + p_{a,j,k} - p_{a,i,k}  \leq 1 \label{eqC8}\\
 \forall i\in X, \forall j \in X : & p_{i,i,j} = 1 \label{eqC9}
\end{eqnarray}

\paragraph{Constraints for Borda.}
We introduce a new integer variable $N_{i,j}$ to represent the number of voters preferring $i$ over $j$ in the weighted tournament. 
We first make sure that the values of $N_{i,j}$ are consistent throughout the weighted tournament. 
\begin{eqnarray}
\forall i\in X, \forall j\in X, \forall k\in X, \forall l\in X : & N_{i,j} + N_{j,i} = N_{k,l} + N_{l,k}
\label{eqB1}
\end{eqnarray}

In each state, when agent $i$ wins, we must make sure that her total amount of points is the highest among all agents in this state (note that $i$ can simply tie with agents it has priority over in the tie-breaking; we omit this for the sake of readability): 
\begin{eqnarray}
&& \forall s\in S,  \forall i \in A(s), \forall k \in A(s) \setminus \{i\} :  \nonumber\\
& &  (1-w_{s,i}) \times M + \sum_{j\in A(s) \setminus \{i\}} {N_{i,j}} > \sum_{j\in A(s) \setminus \{k\}} {N_{k,j}}
\label{eqB2}
\end{eqnarray}
Here $M$ is an arbitrary large value, used to relax the constraint when $w_{s,i}$ is 0.

\end{document}